\documentclass[letterpaper,journal]{IEEEtran}
\usepackage{amsmath,amsfonts}
\usepackage{amssymb}
\usepackage{amsthm}
\usepackage{algorithmic}
\usepackage{algorithm}
\usepackage{array}
\usepackage[caption=false,font=normalsize,labelfont=rm,textfont=rm]{subfig}
\usepackage{textcomp}
\usepackage{stfloats}
\usepackage{url}
\usepackage{verbatim}
\usepackage{graphicx}
\usepackage{listings}
\usepackage{xcolor}
\usepackage{cite}
\usepackage{multirow}
\graphicspath{{fig/}}

\definecolor{codekeyword}{rgb}{0.13,0.29,0.53}
\definecolor{codecomment}{rgb}{0.00,0.50,0.00}
\definecolor{codestring}{rgb}{0.58,0.00,0.82}
\definecolor{codetype}{rgb}{0.15,0.55,0.55}
\definecolor{codenumber}{rgb}{0.50,0.50,0.50}
\lstdefinestyle{cppstyle}{
  language=C++,
  basicstyle=\ttfamily\footnotesize,
  keywordstyle=\color{codekeyword}\bfseries,
  commentstyle=\color{codecomment}\itshape,
  stringstyle=\color{codestring},
  numberstyle=\tiny\color{codenumber},
  emph={NodeId,MapHandle,Entry,StoredKey,StoredValue,%
        uint32_t,uint16_t,int16_t,kInvalidId,kEmptyMapHandle},
  emphstyle=\color{codetype},
  frame=single,
  columns=fullflexible,
  breaklines=true,
  showstringspaces=false,
  keepspaces=true,
}
\providecommand{\llbracket}{\ensuremath{[\![}}
\providecommand{\rrbracket}{\ensuremath{]\!]}}

\newtheorem{theorem}{Theorem}
\newtheorem{lemma}{Lemma}

\newcommand{\Norm}{\operatorname{Norm}}
\newcommand{\Cof}{\operatorname{Cof}}

\newcommand{\ite}{\mathrm{ite}}

\usepackage{hyperref}
\hypersetup{
  pdftitle={XBDD: A Highly Optimized ROBDD with Per-Edge Variable-Flip Maps},
  pdfauthor={Yinglong Gan, Jintao Yu, Shenggang Ying, Yusen Li, Xin Hong}
}

\begin{document}

\title{XBDD: A Highly Optimized ROBDD with Per-Edge Variable-Flip Maps}

\author{Yinglong Gan, Jintao Yu, Shenggang Ying, Yusen Li,
        and~Xin Hong%
\thanks{This work was supported in part
        by the Quantum Science and Technology---National Science and
        Technology Major Project under Grant 2024ZD0300502, and in part by
        the National Natural Science Foundation of China under Grant
        6267073710 and Grant 62272252.
        \textit{(Corresponding authors: Yusen Li; Xin Hong.)}}%
\thanks{Y. Gan and Y. Li are with the College of Computer Science,
        Nankai University, Tianjin 300350, China (e-mail:
        ganyl@nbjl.nankai.edu.cn; liyusen@nbjl.nankai.edu.cn).}%
\thanks{J. Yu is with Arclight Quantum Computing Inc., Beijing, China
        (e-mail: yujt@arclightquantum.com).}%
\thanks{S. Ying and X. Hong are with the Key Laboratory of System Software
        (Chinese Academy of Sciences), Beijing 100190, China, and also with
        the Institute of Software, Chinese Academy of Sciences, Beijing
        100190, China (e-mail: yingsg@ios.ac.cn; hongxin@ios.ac.cn).}}

% \markboth{IEEE Transactions on Computer-Aided Design of Integrated Circuits and Systems}%
% {Gan \MakeLowercase{\textit{et al.}}: XBDD: A Highly Optimized ROBDD with Per-Edge Variable-Flip Maps}

\maketitle

\begin{abstract}
The Reduced Ordered Binary Decision Diagram (ROBDD) is a canonical representation of Boolean functions and is widely used in tasks such as equivalence checking and satisfiability checking of combinational circuits. Classical ROBDD packages greatly improve the efficiency of building ROBDDs through a series of optimization techniques, and compress the node scale of the ROBDD through complement edges. However, existing implementations do not take into account the local polarity differences of isomorphic Boolean functions, and still produce a distinct node for each polarity combination, thereby causing an explosion in the number of nodes. This paper proposes XBDD, a highly optimized ROBDD that, on the basis of fully implementing complement edges and their accompanying engineering techniques, introduces a per-edge variable-flip map. XBDD attaches a flip map to each edge to indicate which input variables must be negated when that edge is followed. This allows nodes that differ only in local input polarities to be merged, further reducing the node count. For certain function families, this sharing even yields exponential compression. We also propose methods that use a bitmap and a map pool to substantially reduce the extra overhead brought by the map, and propose normalization and cofactor operators for the map. In addition, XBDD implements several other engineering optimizations to further improve both time and space efficiency. Experiments show that XBDD trades a controllable time cost for a significant space gain, validating the effectiveness of the per-edge variable-flip map.
\end{abstract}


\section{Introduction}

\IEEEPARstart{E}{fficient} representation and manipulation of Boolean functions are
fundamental to many electronic design automation (EDA) tasks, such as
logic synthesis and formal verification. Reduced Ordered Binary Decision
Diagrams (ROBDDs) provide a canonical graph-based representation under a
fixed variable order \cite{ref_bryant,ref_bryant_survey}. Once constructed,
equivalent functions have identical root references,
enabling constant-time equivalence checking \cite{ref_malik,ref_fujita_cmp}.
Satisfiability is likewise determined by comparing the root with the
constant-0 representation \cite{ref_bryant_heule}. As circuit complexity
grows, the cost of constructing and storing these representations becomes
critical.

Mature ROBDD packages combine several engineering optimizations
\cite{ref_brace}. One important technique among these is called
\emph{complement edges}, which allow a function and its complement to
share the same underlying nodes. Instead of constructing another graph
for $\overline{f}$, a complement edge stores a single bit indicating
whether the output of $f$ should be inverted. In this way, ROBDDs exploit
sharing between identical sub-functions, while complement edges further
extend sharing to complementary sub-functions.

Despite these optimizations, existing ROBDD representations may still miss important opportunities for node sharing. In particular, complement edges capture
differences in \emph{output polarity}, but cannot directly capture
differences in \emph{input polarity}. Consider the following two Boolean
functions:
\begin{equation}
	f_1(a,b)=a\land b,
	\qquad
	f_2(a,b)=\overline{a}\land b.
	\label{eq:intro-example}
\end{equation}
The two functions are neither identical nor complements of each other.
Therefore, conventional ROBDD reduction and complement edges cannot
directly merge their representations. However, their computations are
highly similar:
\begin{equation}
	f_2(a,b)=f_1(\overline{a},b).
\end{equation}
More generally, a sub-function may appear in many input-polarity variants
that are neither identical nor complementary. Representing these variants
separately limits sharing. This motivates encoding input-polarity
differences on edges, just as complement edges encode output polarity.

To address this problem, we propose \textbf{XBDD}, an extension of ROBDD
that associates a \emph{variable-flip map} with each edge. The map records
which input variables should be interpreted in the opposite polarity when
the edge is followed. Consequently, sub-functions that differ only in
selected input polarities can reuse the same underlying nodes, with
their differences represented compactly on the incoming edges. Conceptually,
XBDD extends the sharing capability of complement edges from
\emph{output-polarity variants} to \emph{input-polarity variants}.

Two challenges arise. First, different combinations of nodes, complement
bits, and flip maps may represent the same function. XBDD uses
normalization and map-aware cofactor operations to establish
\emph{constructive uniqueness} under a fixed variable order, preserving
equivalence checking by root-edge comparison for generated representations.

Second, map storage and manipulation may offset the sharing benefits.
XBDD combines bitmap encoding and map interning with compact pointer-free
nodes, level-partitioned unique tables, and adaptive caching and garbage
collection to control this overhead.

Variable-flip maps provide both theoretical and practical benefits. For a class of polarity-sensitive Boolean functions, XBDD reduces the representation size from exponential to linear compared with complement-edge-based ROBDDs. Experimental results on representative Boolean circuit benchmarks further demonstrate its practical effectiveness: compared with CUDD \cite{ref_cudd}, a mature complement-edge-based ROBDD package, XBDD achieves up to 93\% reduction in final node count and 61\% reduction in memory consumption, with construction time ranging from 24\% faster to 41\% slower. These results demonstrate that input-polarity similarity exposes substantial sharing opportunities that are missed by existing ROBDD representations.

The main contributions of this work are summarized as follows:

\begin{itemize}
	
	\item We identify \emph{input-polarity redundancy} as an important source of missed structural sharing in existing ROBDD representations. To exploit this opportunity, we introduce a per-edge variable-flip map that enables sub-functions differing only in input polarities to share the same underlying nodes.
	
	\item We develop normalization and map-aware cofactor operations to integrate variable-flip maps into the conventional ROBDD construction framework, and establish constructive uniqueness under a fixed variable order. We further show that, for a class of polarity-sensitive Boolean functions, XBDD reduces the representation size from exponential to linear.
	
	\item We develop a memory-efficient XBDD implementation that combines compact bitmap encoding and map interning with optimized node, unique-table, and computed-cache structures, keeping the additional memory and runtime overhead low.
	
	\item We evaluate XBDD against CUDD on IWLS'93 circuit benchmarks \cite{ref_iwls93}. The results show substantial reductions in node count and memory consumption with modest runtime overhead, while also validating the effectiveness of map reuse and the predicted exponential-to-linear reduction.
	
\end{itemize}

The remainder of this paper is organized as follows. Section~\ref{sec:background} reviews ROBDD fundamentals, ITE-based construction, and complement edges. Section~\ref{sec:map} presents the XBDD design, including per-edge variable-flip maps, normalization, construction, uniqueness, and the theoretical compression result. Section~\ref{sec:impl} describes the memory-efficient data structures and runtime optimizations used in the implementation. Section~\ref{sec:exp} evaluates XBDD experimentally and analyzes its space--time tradeoff, map interning efficiency, and exponential-compression case study. Section~\ref{sec:related} discusses related decision-diagram representations and implementations, and Section~\ref{sec:conclusion} concludes the paper.

\section{Background and Motivation}
\label{sec:background}

This section first reviews the basic concepts and reduction rules of ROBDDs. We then introduce ITE-based ROBDD construction. Finally, we discuss complement edges and their role in increasing node sharing, which motivates our work.

\subsection{Basic Concepts of ROBDD}
\label{subsec:xobdd}

A Binary Decision Diagram (BDD) represents a Boolean function as a rooted directed acyclic graph. Each non-terminal node is associated with a Boolean variable and tests whether the variable is assigned 0 or 1. Based on the assignment, evaluation proceeds along either the low (0) edge or the high (1) edge. This process continues until reaching a terminal node labeled 0 or 1, which determines the value of the Boolean function.

For example, consider the Boolean function $f(a,b)=a\land b$. Its BDD first checks $a$. If $a=0$, the result must be 0, regardless of
$b$. If $a=1$, the result depends on $b$, so the graph continues to the
node for $b$. In this sense, a BDD evaluates a Boolean function through a
sequence of simple decisions.

The important difference between a BDD and an ordinary decision tree is
that a BDD can \emph{share} repeated parts of the computation. Consider
two branches of a larger Boolean function that eventually need to evaluate
the same sub-function $g$. A decision tree may contain two separate copies
of $g$, whereas a BDD can let both branches point to the same node for $g$.
The graph can therefore be much smaller than the corresponding decision
tree.

An Ordered BDD (OBDD) requires variables to be tested in the same order
along every path. An OBDD is further reduced using two simple rules.
The first is the \emph{elimination rule}: if the low and high edges of a
single node point to the same child, the tested variable cannot affect the
result, so that node is removed and its incoming edges are redirected to the
common child. The second is the \emph{merging rule}: if two distinct nodes
test the same variable and have identical low children and identical high
children, the nodes represent the same sub-function and are replaced by one
shared node. Thus, the first rule removes a redundant variable test within
one node, whereas the second merges duplicate nodes. An OBDD to which neither
rule applies is called a Reduced Ordered BDD (ROBDD).

These reduction rules give ROBDDs an important property: under a fixed
variable order, each Boolean function has a unique ROBDD representation
\cite{ref_bryant}. Therefore, whenever the same sub-function appears
multiple times, all occurrences can share one node. Modern BDD packages
maintain a unique table to detect such equivalent nodes during construction,
and use a computed table to reuse previously computed Boolean operations.

\subsection{ROBDD Construction with ITE}
\label{subsec:ite}

The reduction rules described above define the canonical structure of an
ROBDD. In practice, however, a BDD package must maintain this structure
while constructing new Boolean functions. A widely used approach is to
express Boolean operations through a common recursive primitive, the
\emph{If-Then-Else} (ITE) operator \cite{ref_brace}. For three Boolean functions $F$, $G$,
and $H$, ITE is defined as
\begin{equation}
	\ite(F,G,H)=F\cdot G+\overline{F}\cdot H,
\end{equation}
which returns $G$ when $F=1$ and $H$ when $F=0$.

ITE provides a unified construction interface because common Boolean
operations can be expressed in terms of it, as shown in
Table~\ref{tab:ite}. For example, $F\land G$ can be written as
$\ite(F,G,0)$, while $\overline{F}$ can be written as $\ite(F,0,1)$.
Therefore, rather than implementing a separate recursive procedure for
each Boolean operator, a BDD package can construct most Boolean
expressions through the same ITE engine.

\begin{table}[!ht]
	\caption{Boolean Operations Expressed as ITE\label{tab:ite}}
	\centering
	\renewcommand{\arraystretch}{1.25}
	\begin{tabular}{|c|c|}
		\hline
		Operation & ITE form\\
		\hline
		F AND G & $\ite(F,G,0)$\\
		\hline
		F OR G & $\ite(F,1,G)$\\
		\hline
		F XOR G & $\ite(F,\overline{G},G)$\\
		\hline
		NOT(F) & $\ite(F,0,1)$\\
		\hline
		F NAND G & $\ite(F,\overline{G},1)$\\
		\hline
		F NOR G & $\ite(F,0,\overline{G})$\\
		\hline
	\end{tabular}
\end{table}

The recursive construction of ITE follows from Shannon expansion. Let
$v$ be the highest variable, according to the fixed variable order,
that appears in $F$, $G$, or $H$. We refer to $v$ as the
\emph{top variable}. For a Boolean function $F$, Shannon expansion with
respect to $v$ gives
\begin{equation}
	F=v\cdot F_v+\overline{v}\cdot F_{\overline v},
\end{equation}
where $F_v=F\big|_{v=1}$ and
$F_{\overline v}=F\big|_{v=0}$ denote the cofactors of $F$ with respect
to $v$, respectively. Applying the same decomposition to the
three arguments of ITE gives
\begin{equation}
	\begin{aligned}
		\ite(F,G,H)
		&=F\cdot G+\overline{F}\cdot H\\
		&=v\cdot\bigl(F_vG_v+\overline{F_v}H_v\bigr)
		+\overline{v}\cdot
		\bigl(F_{\overline v}G_{\overline v}
		+\overline{F_{\overline v}}H_{\overline v}\bigr)\\
		&=\ite\bigl(v,\,
		\ite(F_v,G_v,H_v),\,
		\ite(F_{\overline v},G_{\overline v},H_{\overline v})\bigr).
	\end{aligned}
	\label{eq:ite-rec}
\end{equation}

Equation~\eqref{eq:ite-rec} directly gives the recursive construction
procedure. ITE first computes the cofactors of $F$, $G$, and $H$ with
respect to the top variable $v$. It then recursively evaluates the two
branches, producing
\begin{equation}
	T=\ite(F_v,G_v,H_v)
\end{equation}
for $v=1$ and
\begin{equation}
	E=\ite(F_{\overline v},G_{\overline v},H_{\overline v})
\end{equation}
for $v=0$. If $T=E$, the test on $v$ is redundant and the result is
simply $T$. Otherwise, a node labeled by $v$ with children $E$ and $T$
is required.

For a conventional ROBDD, cofactor extraction is particularly simple.
If the root variable of an input node is $v$, its low and high edges
directly give the cofactors for $v=0$ and $v=1$, respectively. If its
root variable lies below $v$ in the variable order, the represented
sub-function is independent of $v$ at the current level, and the same
node is used for both cofactors.

The newly obtained pair $(E,T)$ is then reduced and canonicalized during
node construction. If $E=T$, the node is eliminated by the ROBDD
reduction rule. Otherwise, the unique table is queried using
$(v,E,T)$ as the structural key. An existing node is reused whenever
the same key has already been created; only a previously unseen key
causes a new node to be allocated. In this way, equivalent
sub-functions constructed through different recursive paths converge
to the same ROBDD node. A computed table further caches previously
evaluated ITE calls and avoids repeating the same recursive
computation.

The standard ITE construction algorithm is given in \cite[Fig.~2]{ref_brace}. Terminal
cases and computed-table lookups terminate the recursion whenever the
result is already known. Otherwise, ITE repeatedly performs three key
steps: \emph{cofactor extraction}, \emph{recursive construction}, and
\emph{reduction and canonicalization through the unique table}.

\subsection{Complement Edges in ROBDDs}
\label{subsec:complement}

The ROBDD construction described above maximizes sharing among
\emph{identical} sub-functions: once a node representing a sub-function
has been created, all later occurrences can reuse the same node through
the unique table. However, two sub-functions that differ only in their
output polarity are still treated as different Boolean functions. For
example, if one part of a graph requires a function $g$ while another
requires its complement $\overline{g}$, a conventional ROBDD would in
general need separate representations for the two.

Complement edges extend this sharing capability by moving the output
polarity difference from the graph structure to the edge
\cite{ref_brace}. Instead of constructing a separate graph for
$\overline{g}$, an edge carries an additional complement bit. When the
bit is 0, the target node is interpreted normally as $g$; when it is 1,
the output of the target node is complemented and interpreted as
$\overline{g}$. Therefore, $g$ and $\overline{g}$ can share exactly the
same underlying nodes.

For example, suppose a node represents
\begin{equation}
	g(a,b)=a\land b.
\end{equation}
An ordinary incoming edge refers to $g(a,b)$, whereas a complemented
incoming edge refers to
\begin{equation}
	\overline{g(a,b)}
	=\overline{a\land b}.
\end{equation}
No additional nodes are required for the complemented function.

The benefit of complement edges is that a very small amount of edge
metadata can eliminate potentially large duplicated subgraphs. More
importantly, they illustrate a general principle for increasing
structural sharing in decision diagrams: when two functions have the
same underlying structure but differ by a simple transformation, the
difference can sometimes be encoded on the incoming edge rather than
in separate nodes.

However, complement edges capture only differences in
\emph{output polarity}. They cannot directly merge sub-functions that
have the same underlying structure but differ in the polarity of one or
more \emph{input variables}. This remaining form of redundancy motivates the variable-flip maps introduced in the next section.

\section{XBDD Design}
\label{sec:map}

Complement edges extend ROBDD sharing to functions with opposite
output polarities, but they do not capture structural similarity caused
by differences in input polarity. XBDD addresses this limitation by
encoding input-polarity transformations on edges, allowing a broader
class of structurally similar sub-functions to share the same nodes.

This section presents the design of XBDD. We first define the
per-edge variable-flip map and the semantics of the extended edge
representation. We then introduce normalization and map-aware cofactor
operations, and show how they are integrated into the conventional ITE
construction framework. Next, we establish the constructive uniqueness
of XBDD under a fixed variable order and analyze the node-sharing
benefits enabled by variable-flip maps.

\subsection{Per-Edge Variable-Flip Map}
\label{subsec:edge}

Complement edges improve node sharing by encoding an
\emph{output-polarity} difference on an edge. XBDD follows the same
general principle, but extends it to \emph{input-polarity} differences.
Instead of constructing separate nodes for sub-functions that have the
same Boolean structure but interpret some input variables with opposite
polarities, XBDD stores the polarity transformation on the incoming
edge.

To represent such transformations, each XBDD edge carries a
\emph{variable-flip map}. Let
\begin{equation}
	V=\{x_1,\ldots,x_n\}
\end{equation}
be the Boolean variable set. A flip map $M\subseteq V$ specifies the
variables whose values are interpreted in the opposite polarity when
an edge is followed. For example,
\begin{equation}
	M=\{x_2,x_5\}
\end{equation}
means that $x_2$ and $x_5$ are flipped, while all other variables are
interpreted normally.

An XBDD edge is represented by the triple
\begin{equation}
	e=(u,p,M),
	\label{eq:xbdd-edge}
\end{equation}
where $u$ is the target node, $p\in\{0,1\}$ is the conventional
complement bit, and $M$ is the variable-flip map. The two edge
attributes have complementary roles: $M$ transforms selected
\emph{inputs} before the target node is evaluated, whereas $p$
optionally complements the \emph{output} afterward.

More formally, let $\Phi_u$ denote the Boolean function represented by
node $u$, and let $\chi_M$ be the bit vector whose set bits correspond
to the variables in $M$. The function represented by edge
$e=(u,p,M)$ is
\begin{equation}
	\llbracket e \rrbracket(x)
	=
	p\oplus\Phi_u(x\oplus\chi_M).
	\label{eq:edge-sem}
\end{equation}
Thus, following an XBDD edge can be understood as three conceptual
steps: first flip the variables specified by $M$, then evaluate the
target node $u$, and finally complement the result if $p=1$.

This definition naturally generalizes conventional complement edges.
When $M=\varnothing$, Eq.~\eqref{eq:edge-sem} reduces to
\begin{equation}
	\llbracket (u,p,\varnothing)\rrbracket(x)
	=
	p\oplus\Phi_u(x),
\end{equation}
which is exactly the semantics of a complement edge. XBDD therefore
preserves output-polarity sharing while additionally enabling
input-polarity sharing.

For example, suppose node $u$ represents
\begin{equation}
	\Phi_u(a,b)=a\land b.
\end{equation}
Then the same node can represent several polarity variants through
different incoming edges:
\begin{equation}
	\begin{aligned}
		(u,0,\varnothing)
		&\rightarrow a\land b,\\
		(u,1,\varnothing)
		&\rightarrow \overline{a\land b},\\
		(u,0,\{a\})
		&\rightarrow \overline{a}\land b.
	\end{aligned}
	\label{eq:edge-example}
\end{equation}
The first two forms are already supported by complement-edge ROBDDs.
The third illustrates the additional sharing enabled by XBDD: the
polarity of $a$ changes, but the underlying node structure is reused.

An XBDD object consists of a shared node DAG $\mathcal G$ and a root
edge:
\begin{equation}
	\mathcal X=(\mathcal G,e_{\mathrm{root}}),
	\qquad
	e_{\mathrm{root}}
	=(u_{\mathrm{root}},p_{\mathrm{root}},M_{\mathrm{root}}).
\end{equation}
The root edge has no source node and determines the Boolean function
represented by the XBDD. Different XBDD objects, such as different
outputs of a multi-output circuit, may use different root edges while
sharing the same underlying nodes. Fig.~\ref{fig:struct} illustrates
this organization.

XBDD uses a single terminal node $\top$. The edge
$(\top,0,\varnothing)$ represents constant 1, whereas
$(\top,1,\varnothing)$ represents constant 0. Since a constant does
not depend on any input variable, attaching a variable-flip map to a
terminal edge has no semantic effect. Terminal edges are therefore
normalized as
\begin{equation}
	(\top,p,M)
	\rightarrow
	(\top,p,\varnothing).
	\label{eq:tcanon}
\end{equation}

To further reduce the number of equivalent structural forms, XBDD
uses a normalized internal-node representation. For every nonterminal
node, the low edge is required to be regular and carries neither a
complement bit nor a flip map. A node whose decision variable is $v$
therefore has the form
\begin{equation}
	u=(v,l,h,q,D),
	\label{eq:node-structure}
\end{equation}
where $l$ and $h$ are the low and high child nodes, respectively, and
$(q,D)$ are the complement bit and variable-flip map associated with
the high edge. Equivalently, the low edge is
$(l,0,\varnothing)$ and the high edge is $(h,q,D)$.

This asymmetric representation does not reduce expressiveness: common
complement and flip information from the two outgoing edges can be
moved to the incoming edge of the node. More importantly, it gives XBDD
a standard structural form that can be used as the key of the unique
table. The next subsection describes how XBDD transforms arbitrary
pairs of cofactors into this normalized form and how the corresponding
flip information is recovered during recursive construction.

\begin{figure}[!t]
	\centering
	\includegraphics[width=0.4\columnwidth]{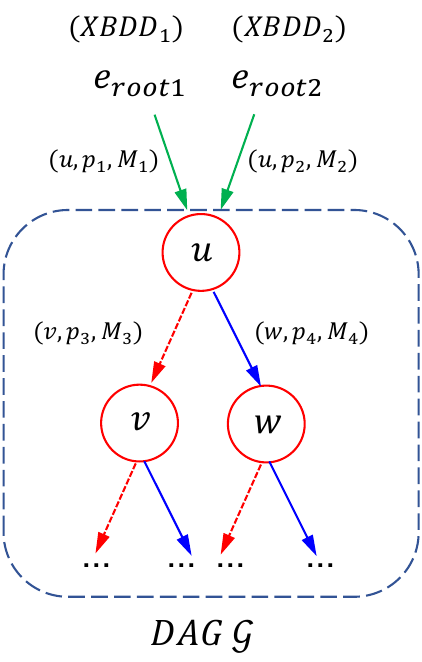}
	\caption{Structure of XBDD objects. Each XBDD object maintains its own
		root edge, while multiple root edges may reference nodes in the same
		shared DAG $\mathcal G$. Every edge is represented by a target node, a
		complement bit, and a variable-flip map.}
	\label{fig:struct}
\end{figure}

\subsection{Normalization and Cofactor}
\label{subsec:norm}

Introducing variable-flip maps increases the number of ways in which
the same underlying Boolean structure can be represented. For example,
two Shannon branches may appear in opposite orders because the polarity
of the current variable is reversed, or the two outgoing edges may
carry common complement or flip information that can be moved to the
incoming edge. If these equivalent forms are inserted into the unique
table directly, structurally equivalent sub-functions may still be
stored as different nodes.

XBDD therefore normalizes every newly constructed node before the
unique-table lookup. The goal of normalization is to separate the
\emph{shared structure} of a node from the transformations applied to
that structure. Branch ordering is fixed deterministically, while
common output-polarity and input-polarity transformations are lifted
from the outgoing edges to the incoming edge. As a result, functions
that differ only by such transformations can reach the same structural
key and reuse the same node.

We denote this normalization operator by $\Norm$. Suppose the low and
high cofactors of variable $v$ are represented by
\begin{equation}
	e_0=(l,p_l,M_l),
	\qquad
	e_1=(h,p_h,M_h).
\end{equation}
$\Norm(v,e_0,e_1)$ constructs an XBDD representation of the Shannon
combination
\begin{equation}
	F=
	\overline{x_v}\llbracket e_0\rrbracket
	+
	x_v\llbracket e_1\rrbracket.
	\label{eq:shannon-comb}
\end{equation}
The child nodes $l$ and $h$ are assumed to lie below $v$ in the
variable order, and $v\notin M_l\cup M_h$. This condition is naturally
satisfied during recursive construction because the cofactor operation
removes the current variable from the incoming map.

Conceptually, normalization performs four main transformations: \emph{redundancy elimination}, \emph{branch ordering}, \emph{attribute lifting}, and
\emph{structural interning}.

\paragraph{Redundancy elimination}
If the two input edges are identical, i.e., $e_0=e_1$, the Shannon test
on $v$ is redundant and $\Norm$ simply returns $e_0$.

\paragraph{Branch ordering}
Input-polarity reversal of the current variable exchanges the two
Shannon branches. To ensure that both branch orders lead to the same
node structure, XBDD imposes a deterministic order on the two child
nodes. Each node is assigned a unique identifier, and the terminal
$\top$ is treated as having the smallest identifier. If the low child
has a larger identifier than the high child, the two edges are swapped.
Let
\begin{equation}
	s=
	\begin{cases}
		1, & \text{if the two branches are swapped},\\
		0, & \text{otherwise}.
	\end{cases}
\end{equation}

Swapping the branches changes the interpretation of the current
variable. XBDD restores the original semantics in
\hyperref[step:attribute-lifting]{Step~c} according to the recorded value
of $s$. In this way, two functions that differ only in the polarity of
$v$ can share the same ordered pair of child nodes.

\paragraph{Attribute lifting}
\label{step:attribute-lifting}
After the branches have been ordered, XBDD removes transformations that
are common to both outgoing edges and moves them to the incoming edge.
For complement bits, the complement bit of the low edge is lifted:
\begin{equation}
	p_{\mathrm{out}}=p_l,
	\qquad
	q=p_h\oplus p_l.
\end{equation}
The low edge therefore becomes regular, while the high edge stores only
the complement difference $q$ relative to the low edge.

The same principle is applied to variable-flip maps. In the general
case,
\begin{equation}
	M_{\mathrm{out}}=M_l,
	\qquad
	D=M_l\triangle M_h,
	\label{eq:common-map}
\end{equation}
where $\triangle$ denotes symmetric difference. Thus, the incoming edge
carries the transformation of the low branch, while the high edge keeps
only the map difference relative to it.

When the ordered low child is the terminal node $\top$ and the high
child is nonterminal, the low branch is independent of all variables.
Its map therefore provides no meaningful reference transformation. In
this case, XBDD instead lifts the map of the high branch:
\begin{equation}
	M_{\mathrm{out}}=M_h,
	\qquad
	D=\varnothing.
	\label{eq:const-lift}
\end{equation}
We refer to this rule as \emph{constant low-child lifting}.

Finally, if the branches were swapped during branch ordering, i.e.,
$s=1$, the current variable is added to the incoming map:
\begin{equation}
	M_{\mathrm{out}}
	\leftarrow
	M_{\mathrm{out}}\cup\{v\}.
\end{equation}
This records the branch reversal as an input-polarity transformation
rather than as a different node structure.

\paragraph{Structural interning}
After normalization, the low edge has the canonical form
$(l,0,\varnothing)$ and the high edge has the form $(h,q,D)$. The node
can therefore be identified by the structural key
\begin{equation}
	(v,l,h,q,D).
	\label{eq:struct-key}
\end{equation}
XBDD looks up this key in the unique table. If an identical key already
exists, the corresponding node is reused; otherwise, a new node is
created. The normalized result is returned through the incoming edge
\begin{equation}
	e_{\mathrm{out}}
	=
	(u,p_{\mathrm{out}},M_{\mathrm{out}}).
\end{equation}

The normalization procedure preserves the represented Boolean
function.

\begin{theorem}[Semantics preservation]
	\label{thm:norm-semantics}
	For any valid inputs $e_0$ and $e_1$,
	\begin{equation}
		\llbracket\Norm(v,e_0,e_1)\rrbracket
		=
		\overline{x_v}\llbracket e_0\rrbracket
		+
		x_v\llbracket e_1\rrbracket.
	\end{equation}
\end{theorem}

\begin{proof}
	Lifting the same complement transformation from both branches to the
	incoming edge does not change the represented function. Similarly,
	$M_l\triangle M_h$ records exactly the input-polarity difference
	remaining after the transformation $M_l$ has been lifted to the
	incoming edge. If the two branches are exchanged, adding $v$ to the
	incoming flip map reverses the interpretation of $x_v$ and therefore
	restores the original Shannon decomposition. In the constant
	low-child case, the terminal branch is independent of all input
	variables, so lifting the high-branch map leaves the semantics
	unchanged.
\end{proof}

\medskip
\noindent\textbf{Map-aware cofactor.}
Normalization combines two cofactors into one normalized XBDD edge.
Recursive construction also requires the reverse operation: given an
edge $e=(u,p,M)$, we must obtain the sub-function produced by fixing the
current variable $x_v$ to either 0 or 1. We denote this operation by
$\Cof(e,v,c)$, where $c\in\{0,1\}$.

The important difference from a conventional ROBDD is that the incoming
flip map may reverse which outgoing branch corresponds to a given value
of $x_v$. Let
\begin{equation}
	b=[v\in M].
\end{equation}
When computing the cofactor for $x_v=c$, the branch seen by node $u$ is
therefore
\begin{equation}
	j=c\oplus b.
	\label{eq:branch-selection}
\end{equation}
Thus, when $v\notin M$, the conventional branch $c$ is selected; when
$v\in M$, the low and high branches are exchanged.

If the root variable of $u$ lies below $v$, then the node itself does
not test $x_v$. In this case, fixing $x_v$ only consumes the
corresponding flip information, and
\begin{equation}
	\Cof(e,v,c)
	=
	(u,p,M\setminus\{v\}).
\end{equation}
If $u=\top$, the terminal normalization rule further clears the map,
yielding $(\top,p,\varnothing)$.

Now consider the case where the root variable of $u$ is exactly $v$.
Recall that the normalized node has low edge
$(l,0,\varnothing)$ and high edge $(h,q,D)$. Let $u_0=l$ and $u_1=h$.
After selecting branch $j$ according to
Eq.~\eqref{eq:branch-selection}, the cofactor is
\begin{equation}
	\Cof(e,v,c)
	=
	\left(
	u_j,\;
	p\oplus [j=1]q,\;
	\bigl(M\triangle [j=1]D\bigr)\setminus\{v\}
	\right),
	\label{eq:cof}
\end{equation}
where $[j=1]q$ denotes $q$ when the high branch is selected and
$0$ otherwise, while $[j=1]D$ denotes $D$ when the high branch is
selected and $\varnothing$ otherwise. If the selected child is the terminal node,
its map is again cleared according to Eq.~\eqref{eq:tcanon}.

Equation~\eqref{eq:cof} can be interpreted in three steps. First, the
incoming map determines whether the requested cofactor selects the low
or high branch. Second, if the high branch is taken, its relative
complement bit and flip map are composed with those of the incoming
edge. Finally, $v$ is removed from the map because its value has already
been fixed.

The operators $\Norm$ and $\Cof$ therefore play complementary roles.
$\Norm$ combines two cofactors into a normalized shared structure,
while $\Cof$ recovers the corresponding branches during recursive
construction. This relationship is the basis of the XITE algorithm
introduced in the next subsection.

\subsection{XBDD Construction}
\label{subsec:construct}

With the edge representation and the two operations $\Norm$ and $\Cof$
defined above, XBDD can be constructed within the same recursive
framework as a conventional ROBDD. In particular, XBDD retains the ITE
operator as its common construction primitive and expresses Boolean
operations using the same ITE forms listed in Table~\ref{tab:ite}.

The key difference is that the two operations in conventional ITE that
directly manipulate BDD edges must now account for variable-flip maps.
First, ordinary cofactor extraction is replaced by the map-aware
operator $\Cof$, which determines the correct branch after considering
the incoming flip map and composes the edge attributes accordingly.
Second, the two recursively constructed branches are not inserted into
the unique table directly. Instead, they are first passed to $\Norm$,
which removes redundant transformations, imposes a canonical branch
ordering, and returns a normalized edge whose underlying node can be
shared through the unique table.

In other words, the recursive structure of ITE remains unchanged. We refer to the ITE procedure specialized for XBDD as $\textsc{XITE}$. Algorithm~\ref{alg:xite} summarizes the construction. Given three XBDD edges $F$, $G$, and $H$, XITE first handles terminal cases and checks
the computed table, as in conventional ITE. Otherwise, it selects the
top variable $v$ among the three arguments. In XBDD, variables at higher
levels are closer to the root; thus, selecting the top variable amounts
to selecting the variable with the largest level. The map-aware cofactors for $v=1$ and $v=0$ are then computed using $\Cof$, and the two resulting subproblems are solved recursively. Finally, the two recursive results are combined by $\Norm(v,E,T)$. Here, both the $\Cof$ and $\Norm$ operators are defined in
Section~\ref{subsec:norm}.

\begin{algorithm}[!t]
	\caption{$\textsc{XITE}(F,G,H)$}
	\label{alg:xite}
	\begin{algorithmic}[1]
		\IF{terminal case}
		\STATE \textbf{return} result
		\ELSIF{computed-table has entry $\{F,G,H\}$}
		\STATE \textbf{return} cached result
		\ELSE
		\STATE let $v$ be the top variable of $\{F,G,H\}$
		\STATE $F_v,G_v,H_v
		\gets \Cof(F,v,1),\Cof(G,v,1),\Cof(H,v,1)$
		\STATE $F_{\overline v},G_{\overline v},H_{\overline v}
		\gets \Cof(F,v,0),\Cof(G,v,0),\Cof(H,v,0)$
		\STATE $T \gets \textsc{XITE}(F_v,G_v,H_v)$
		\STATE $E \gets
		\textsc{XITE}(F_{\overline v},
		G_{\overline v},
		H_{\overline v})$
		\STATE $R \gets \Norm(v,E,T)$
		\label{line:xite-norm}
		\STATE $\textsc{insert\_computed\_table}(\{F,G,H\},R)$
		\STATE \textbf{return} $R$
		\ENDIF
	\end{algorithmic}
\end{algorithm}

XBDD does not require a separate construction algorithm for each Boolean operator. AND, OR, XOR, NOT, and other operators are translated into XITE calls in exactly the same way that they are translated into ITE calls for conventional ROBDDs. Thus, the variable-flip map extends the representation and the local edge operations without changing the high-level Boolean construction interface.

The normalization performed at the end of every nonterminal XITE call
also provides the foundation for representation uniqueness. In the
next subsection, we show that, under a fixed variable order, edges
generated through this construction procedure have a unique
representation for their Boolean functions.

\subsection{Constructive Uniqueness}
\label{subsec:canonical}

A key property of conventional ROBDDs is canonicity: under a fixed
variable order, each Boolean function has a unique reduced
representation. Variable-flip maps make this property more subtle.
Because an XBDD edge contains a target node, a complement bit, and a
flip map, the same Boolean function may admit multiple syntactically
valid edge triples.

XBDD therefore does not claim \emph{construction-independent strong
	canonicity} over all possible well-formed XBDD edges. Instead, it
guarantees \emph{constructive uniqueness}: every Boolean function built
through the XITE construction procedure has a unique generated edge.
This guarantee is sufficient for normal package use, because all
Boolean expressions exposed through the XBDD interface are constructed
through XITE. Consequently, equivalence between two constructed
functions can still be determined by comparing their root edges.

Let $\mathcal G$ denote the set of all edges generated by XBDD
construction. It contains the normalized constant and variable edges,
all results returned by XITE, and all intermediate edges returned by
recursive XITE calls. The uniqueness result below is stated only for
edges in $\mathcal G$.

The proof relies on one central property of the construction:
normalization and cofactor extraction are mutually consistent. In
particular, a generated nonterminal edge can be reconstructed uniquely
from its two cofactors.

\begin{lemma}[Reconstruction of a generated edge]
	\label{lem:generated-reconstruction}
	Let $R\in\mathcal G$ be a nonterminal edge whose root variable is $v$.
	Define
	\begin{equation}
		R_c=\Cof(R,v,c),
		\qquad c\in\{0,1\}.
	\end{equation}
	Then $R_0,R_1\in\mathcal G$, and
	\begin{equation}
		R=\Norm(v,R_0,R_1).
		\label{eq:generated-reconstruction}
	\end{equation}
\end{lemma}

\begin{proof}
	A nonterminal edge generated by XITE is obtained by recursively
	constructing two branch results and then applying
	$\Norm(v,\cdot,\cdot)$. By the definition of $\Cof$, applying
	$\Cof(R,v,0)$ and $\Cof(R,v,1)$ reverses the transformations introduced
	by normalization: it restores any branch exchange, composes the
	relative complement and flip attributes, and removes the current
	variable from the map. Therefore, the two cofactors recovered from $R$
	are exactly the two recursive branch results used to construct it.
	
	Applying $\Norm$ to these recovered cofactors performs the same
	deterministic branch ordering, attribute lifting, and unique-table
	lookup as in the original construction, and therefore returns the same
	generated edge $R$.
\end{proof}

We can now establish uniqueness for all generated XBDD edges.

\begin{theorem}[Constructive uniqueness of XBDD]
	\label{thm:constructive-uniqueness}
	Under a fixed variable order, for any two generated edges
	$R,S\in\mathcal G$,
	\begin{equation}
		R=S
		\quad\Longleftrightarrow\quad
		\llbracket R\rrbracket
		=
		\llbracket S\rrbracket.
		\label{eq:constructive-uniqueness}
	\end{equation}
\end{theorem}

\begin{proof}
	The forward direction is immediate: identical edges have identical
	semantics.
	
	For the reverse direction, suppose that $R$ and $S$ represent the same
	Boolean function. We prove, by induction on the variable level $k$, that
	the claim holds for all generated edges whose root variables lie at or
	below level $k$, with terminal edges treated as lying below all variable
	levels.
	
	For the base case, both edges are terminal. XBDD uses a single terminal
	node $\top$, and terminal normalization removes all flip-map information.
	Hence constant 1 is represented only by
	$(\top,0,\varnothing)$ and constant 0 only by
	$(\top,1,\varnothing)$. Therefore, two generated edges representing the
	same constant must be identical.
	
	Now assume that the statement holds below level $k$, and consider
	generated edges $R$ and $S$ whose root variables lie at or below level
	$k$. If both root variables lie below level $k$, the result follows
	directly from the induction hypothesis.
	
	Suppose next that exactly one root variable, say that of $R$, is at level
	$k$, and let it be $v$. Since the root variable of $S$ lies below $v$,
	and flip maps do not introduce variable dependencies,
	$\llbracket S\rrbracket$ is independent of $v$. Let
	\begin{equation}
		R_c=\Cof(R,v,c),
		\qquad c\in\{0,1\}.
	\end{equation}
	By Lemma~\ref{lem:generated-reconstruction}, $R_0,R_1\in\mathcal G$ and
	their root variables lie below level $k$. Because $R$ and $S$ have the
	same semantics and $S$ is independent of $v$,
	\begin{equation}
		\llbracket R_0\rrbracket=\llbracket R_1\rrbracket.
	\end{equation}
	The induction hypothesis gives $R_0=R_1$. The redundancy-elimination
	rule of $\Norm$ and Lemma~\ref{lem:generated-reconstruction} then imply
	$R=\Norm(v,R_0,R_1)=R_0$, contradicting that the root variable of $R$ is
	at level $k$. Hence this case is impossible.
	
	It remains to consider the case where both root variables are at level
	$k$. Under the fixed variable order, both are the same variable $v$.
	Define
	\begin{equation}
		R_c=\Cof(R,v,c),
		\qquad
		S_c=\Cof(S,v,c),
		\qquad c\in\{0,1\}.
	\end{equation}
	By Lemma~\ref{lem:generated-reconstruction}, these cofactors belong to
	$\mathcal G$ and have root variables below level $k$. Since $R$ and $S$
	have the same semantics, the semantics of $\Cof$ gives
	\begin{equation}
		\llbracket R_c\rrbracket
		=
		\llbracket S_c\rrbracket,
		\qquad c\in\{0,1\}.
	\end{equation}
	By the induction hypothesis, $R_0=S_0$ and $R_1=S_1$. Using
	Lemma~\ref{lem:generated-reconstruction},
	\begin{equation}
		\begin{aligned}
			R
			&=\Norm(v,R_0,R_1)\\
			&=\Norm(v,S_0,S_1)\\
			&=S.
		\end{aligned}
	\end{equation}
	Hence two generated edges represent the same Boolean function if and
	only if they are identical.
\end{proof}

The theorem directly preserves the main practical benefit of ROBDD
canonicity. A user expression is translated into a sequence of XITE
calls, so its result belongs to $\mathcal G$. Therefore, once two
functions have been constructed, checking whether they are equivalent
requires only a comparison of their root edges.

Constructive uniqueness is intentionally weaker than requiring every
syntactically valid XBDD edge for a function to be identical. For
example, some node functions may possess input/output polarity
symmetries, allowing different combinations of complement bits and flip
maps to denote the same function. Eliminating all such equivalent edge
forms would require detecting these symmetries during node
construction and selecting a globally preferred representative. Such
checks would make normalization substantially more expensive and would
be placed on the critical path of every XITE call.

XBDD instead uses lightweight local normalization and guarantees
uniqueness for all representations produced through its construction
interface. This design preserves constant-time equivalence checking for
normally constructed XBDD objects while avoiding expensive global
symmetry analysis.

\subsection{Node Compression with Variable-Flip Maps}
\label{subsec:compression}

Variable-flip maps extend node sharing beyond identical and complementary
sub-functions. In particular, sub-functions that differ only in the
polarities of selected input variables can reuse the same underlying
nodes, with their differences encoded in the incoming edge maps. We first
illustrate this mechanism with a small example and then show that the
resulting compression can become asymptotically significant.

Consider the variable order $x_3 \succ x_2 \succ x_1$ and the two
functions
\begin{equation}
	f_1=x_1x_2x_3,\qquad
	f_2=x_1\overline{x_2}x_3,
	\label{eq:ex-f}
\end{equation}
which differ only in the polarity of $x_2$. Splitting them on $x_3$
produces the high cofactors $x_1x_2$ and
$x_1\overline{x_2}$. These two cofactors are neither identical nor
complementary and therefore require distinct nodes in an ROBDD with
complement edges. In XBDD, however, normalization moves their
input-polarity difference to the incoming edge map. Both cofactors
therefore share the same underlying node $h$:
\begin{equation}
	\mathcal C(x_1x_2)
	=(h,1,\varnothing),\qquad
	\mathcal C(x_1\overline{x_2})
	=(h,1,\{x_2\}).
	\label{eq:ex-cof}
\end{equation}
The structural representation is identical; only the edge map records
that $x_2$ is interpreted with the opposite polarity in the second
cofactor.

The same sharing propagates to the next level. Since both $f_1$ and
$f_2$ have the constant-$0$ low cofactor with respect to $x_3$,
normalizing the two pairs of children again produces the same underlying
node $u$. Consequently,
\begin{equation}
	\mathcal C(f_1)
	=(u,1,\varnothing),\qquad
	\mathcal C(f_2)
	=(u,1,\{x_2\}).
	\label{eq:ex-root}
\end{equation}
Thus, a polarity difference that would otherwise require distinct
subgraphs is represented entirely by an edge attribute in XBDD.
Fig.~\ref{fig:compare} illustrates this additional sharing.

\begin{figure}[!t]
	\centering
	\subfloat[]{%
		\includegraphics[width=0.48\columnwidth]{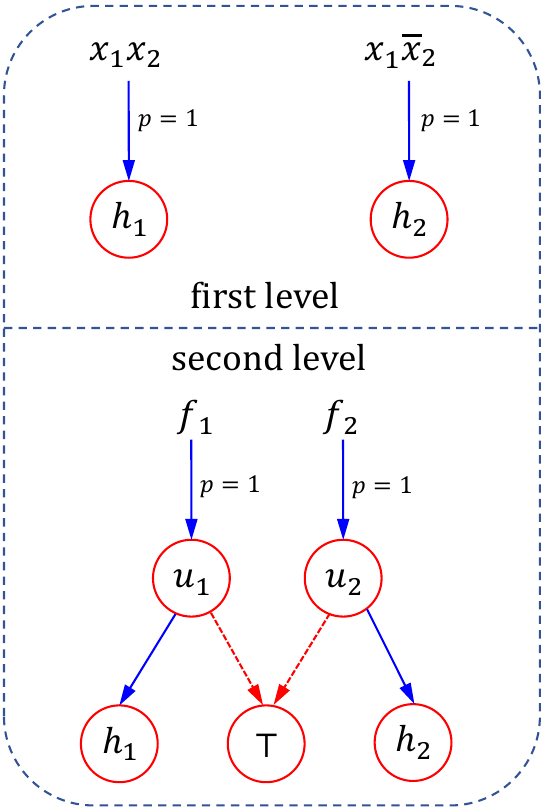}%
		\label{fig:robdd}}
	\hfill
	\subfloat[]{%
		\includegraphics[width=0.48\columnwidth]{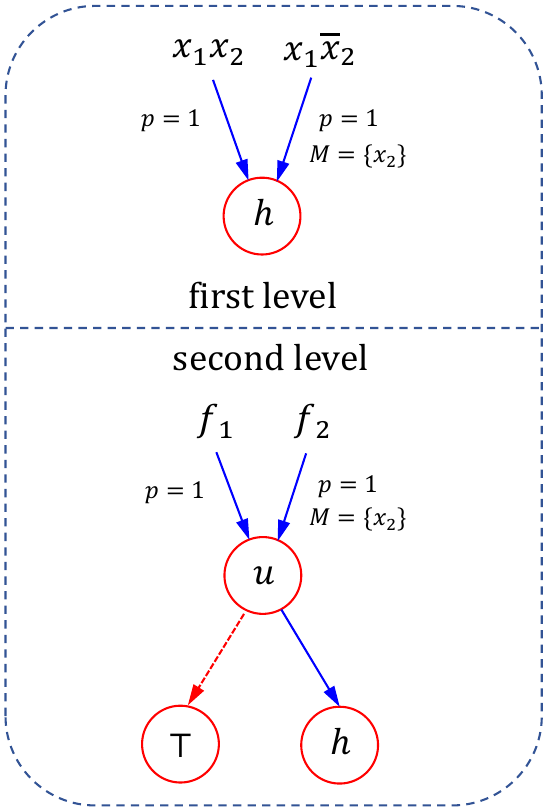}%
		\label{fig:xbdd}}
	\caption{Node sharing for functions with different input polarities:
		(a) an ROBDD with complement edges requires separate nodes for
		the polarity variants, whereas (b) XBDD records the polarity
		difference in an edge map and allows the underlying nodes to be
		shared.}
	\label{fig:compare}
\end{figure}

Although the example above saves only a few nodes, the same mechanism can
repeat across many variables and lead to an asymptotic reduction in
representation size. Consider the following family of Boolean functions:
\begin{equation}
	F_k(c_1,\ldots,c_k,x_1,\ldots,x_k)
	=
	\bigwedge_{j=1}^{k}(x_j\oplus c_j),
	\label{eq:exp-family}
\end{equation}
under the variable order
\begin{equation}
	c_1\succ c_2\succ\cdots\succ c_k
	\succ x_1\succ x_2\succ\cdots\succ x_k.
	\label{eq:exp-order}
\end{equation}
Here, each $c_j$ determines the polarity in which $x_j$ is interpreted.
After assigning the $k$ variables $c_1,\ldots,c_k$, the remaining
function takes the form
\begin{equation}
	\bigwedge_{j=1}^{k}(x_j\oplus b_j),
	\qquad
	(b_1,\ldots,b_k)\in\{0,1\}^{k}.
	\label{eq:polarity-variants}
\end{equation}
Hence, the upper part of the diagram can expose $2^k$ different
polarity variants of the same conjunction structure.

\medskip
\noindent\textbf{Proposition 1.}
Under the variable order in \eqref{eq:exp-order}, a complement-edge
ROBDD for $F_k$ requires $\Omega(2^k)$ nodes, whereas $F_k$ admits an
XBDD representation with $O(k)$ nodes.

\medskip
\noindent\textit{Proof sketch.}
For every assignment $(b_1,\ldots,b_k)$ to the $c$ variables,
the corresponding residual function is the conjunction in
\eqref{eq:polarity-variants}. These $2^k$ residual functions are
pairwise distinct. Moreover, for $k>1$, no two such conjunctions are
complements of each other. Complement edges therefore cannot collapse
them into a common representation, yielding exponentially many
distinct residual structures.

In XBDD, all these residual functions share the same conjunction
structure over $x_1,\ldots,x_k$. Their only differences are the input
polarities specified by $(b_1,\ldots,b_k)$, which can be encoded in the
variable-flip maps of their incoming edges. The shared conjunction
requires only a linear number of nodes. The decisions on
$c_1,\ldots,c_k$ likewise require only a linear number of nodes, with
their effects propagated through edge maps. Therefore, the complete
XBDD representation contains $O(k)$ nodes.
\hfill$\square$

This family demonstrates that variable-flip maps do more than provide
local node savings: when polarity variants repeatedly occur across
multiple levels, they can reduce an exponential number of structurally
similar sub-functions to a linear-size shared representation.
Section~\ref{subsec:expcompress} experimentally validates this
exponential-to-linear behavior. For general circuits, the achievable
reduction naturally depends on how frequently sub-functions differ
primarily in input polarity.

\section{Implementation}
\label{sec:impl}

The variable-flip map enables XBDD to exploit additional structural
sharing, but it also introduces extra metadata and map operations during
XBDD construction. A practical implementation must therefore ensure that
the cost of maintaining these transformations does not offset the memory
savings obtained from node reduction. To this end, XBDD combines compact
map encoding, memory-efficient core data structures, and several runtime
optimizations.

\subsection{Compact Variable-Flip Map Representation}
\label{subsec:map-impl}

Variable-flip maps are frequently manipulated during normalization and
cofactor computation and may appear in a large number of nodes and cache
entries. Directly storing a set of variable identifiers on each edge
would therefore introduce considerable storage and lookup overhead. XBDD
addresses this issue through bitmap encoding and map interning.

\medskip
\noindent\textbf{Bitmap encoding.}
Since variable identifiers occupy a contiguous range $[0,n)$, a flip map
can be represented naturally as a bitmap. For $n$ variables, XBDD stores
it in $w=\left\lceil\frac{n}{64}\right\rceil$ machine words of 64 bits
each, where bit $i$ is set if and only if variable $x_i$ belongs to the
map.

This representation makes the common map operations inexpensive.
Symmetric difference, which is heavily used by normalization and
cofactoring, becomes a word-wise XOR. Equality is implemented by
word-wise comparison, whereas insertion, deletion, and membership testing
require only a bit operation on the corresponding word. The costs of
these operations are summarized in Table~\ref{tab:bitmap}.

\begin{table}[!t]
	\caption{Common operations on the bitmap representation.}
	\label{tab:bitmap}
	\centering
	\begin{tabular}{lcc}
		\hline
		Operation & Implementation & Complexity\\
		\hline
		$A\triangle B$ & word-wise XOR & $O(w)$\\
		Equality & word-wise comparison & $O(w)$\\
		Insert/erase/test & bit operation & $O(1)$\\
		Hash & word-wise folding & $O(w)$\\
		\hline
	\end{tabular}
\end{table}

When $n\leq64$, the entire map fits into one machine word, making the
major map operations effectively constant-time.

\medskip
\noindent\textbf{Map interning.}
Although bitmap encoding makes individual maps compact, repeatedly
storing the same bitmap on many edges would still incur substantial
memory overhead. In practice, many edges reuse the same flip patterns.
XBDD therefore maintains a global \emph{map pool}, where each distinct
bitmap is stored once and assigned a 32-bit handle $\iota(M)$, with
$\iota(\varnothing)=0$.

Nodes and computed-cache entries store only this handle rather than the
complete bitmap. Besides reducing duplicate storage, interning also turns
map equality into an integer comparison:
\begin{equation}
	M_1=M_2
	\iff
	\iota(M_1)=\iota(M_2).
\end{equation}
The empty map is handled as a common fast path. Since its handle is zero,
an edge without an input-polarity transformation can be identified
without accessing the map pool.

This representation is especially useful because maps occur not only in
nodes but also in computed-cache keys and results. Section~\ref{subsec:intern}
evaluates the map reuse rate and the effectiveness of interning in
practice.

\subsection{Memory-Efficient Core Data Structures}
\label{subsec:memory_impl}

XBDD construction may generate millions of nodes and hash-table entries,
so even a small increase in per-object storage can lead to substantial
memory overhead. XBDD therefore organizes its node pool and unique table
around compact fixed-width representations.

\medskip
\noindent\textbf{Compact node representation.}
All nodes are stored in a contiguous node pool and referenced using
32-bit \texttt{NodeId}s rather than 64-bit memory pointers. The terminal
node is fixed at index 0, making terminal detection a simple integer
comparison.

The complement flag is packed into the high bit of the child identifier,
so the target node and output polarity occupy a single 32-bit field.
Moreover, the normalized node representation described in
Section~\ref{subsec:norm} guarantees that the low edge carries an empty
variable-flip map. Therefore, an internal node stores only one 32-bit map
handle for the relative high-edge map.

Together with a 16-bit reference counter and a 16-bit variable-level
field, each XBDD node occupies 16 bytes. The reference counter is
saturating: once the maximum value is reached, the node is conservatively
treated as permanently referenced rather than allowing the counter to
overflow.

An important consequence is that normalization serves not only to obtain
a deterministic representation but also to reduce implementation cost.
Without the normalized low edge, each node would require two independent
map handles instead of one.

\medskip
\noindent\textbf{Level-partitioned unique table.}
As in conventional ROBDD packages, XBDD uses a unique table to ensure
that structurally identical normalized nodes are stored only once. Rather
than maintaining a single global table, XBDD partitions the unique table
by variable level.

Each level owns an independent open-addressing hash table. For a fixed
level, a node is identified by
\begin{equation}
	(low,\;high,\;q,\;\iota(D)),
	\label{eq:unique-key}
\end{equation}
where $low$ and $high$ are child identifiers, $q$ is the relative
complement bit, and $\iota(D)$ is the handle of the relative high-edge
map. Since the variable level is implicit in the selected partition, it
does not need to be included in the lookup key.

Partitioning also improves implementation efficiency. Table growth and
rehashing affect only one level rather than the entire unique table, and
recursive XBDD construction often accesses nodes from the same or nearby
levels consecutively, providing better memory locality.

\subsection{Runtime Optimizations}
\label{subsec:runtime_impl}

The compact representations above control the memory overhead introduced
by variable-flip maps. Runtime overhead, however, may still arise from
memoization, recursive Boolean operations, and reclamation of temporary
nodes. XBDD therefore applies several optimizations to these common
execution paths.

\medskip
\noindent\textbf{Compact and adaptive computed cache.}
Recursive Boolean operations repeatedly encounter identical
sub-problems, making memoization essential for efficient construction.
XBDD represents each edge in the cache using a 32-bit node/complement
field and a 32-bit map handle. Consequently, an ITE cache entry containing
three input edges and one result edge occupies 32 bytes.

The cache is direct-mapped: each key hashes to one slot, and a collision
simply replaces the existing entry. This organization avoids pointer
chasing and complex replacement policies. Its capacity is adjusted
according to the scale of the current construction and cache behavior, so
small instances do not reserve excessive memory while larger instances
can benefit from a larger memoization space.

To reduce random memory accesses on cache misses, XBDD additionally
maintains a one-byte partial-hash tag for each cache slot. A lookup first
checks this tag; a mismatch immediately indicates a miss without loading
the complete 32-byte entry. Only a matching tag triggers full-key
comparison. Since the tag array is much smaller than the main cache, it
has better cache locality and acts as a lightweight filter for
miss-dominated accesses.

\medskip
\noindent\textbf{Adaptive garbage collection.}
XBDD manages node lifetime using reference counting. When the reference
count of a node reaches zero, the node becomes dead but is not immediately
removed from the unique table. Retaining recently dead nodes allows them
to be reused if the same intermediate sub-function reappears shortly
afterward, avoiding unnecessary reconstruction.

Garbage collection is triggered when the number of dead nodes exceeds an
adaptive threshold. The threshold is adjusted according to the
effectiveness of previous collections: ineffective collections postpone
the next cleanup, while collections that reclaim a large fraction of
nodes allow more aggressive reclamation. The threshold also scales with
the number of live nodes to prevent excessive collection frequency for
large XBDDs.

The map pool is reclaimed together with node and cache cleanup. Maps
referenced by surviving nodes or valid cache entries are retained, while
unreferenced maps can be removed without requiring a separate full
reclamation pass.

\medskip
\noindent\textbf{Skipping transient cache entries.}
Not every intermediate recursive result benefits from memoization. Some
operands exist only along the current recursion path and are unlikely to
be reused. Caching such sub-problems introduces additional hashing and
memory accesses while potentially replacing more useful entries.

XBDD uses reference counts as a lightweight heuristic for identifying
these cases. When all operands of an operation have reference count one,
the operation is treated as transient and bypasses both cache lookup and
cache insertion. This reduces cache traffic for short-lived intermediate
states and preserves cache capacity for more reusable sub-functions.

Overall, these optimizations are designed to keep the additional
metadata and runtime cost of variable-flip maps small. Bitmap encoding
and interning compact the map representation, fixed-width node and
unique-table layouts limit memory overhead, and cache and recursion
optimizations reduce the additional work introduced during construction.
The resulting space--time tradeoff is evaluated in
Section~\ref{sec:exp}.

\section{Evaluation}
\label{sec:exp}

We evaluate XBDD from four perspectives. First, we examine whether the additional sharing enabled by variable-flip maps reduces node count and memory consumption on practical circuits. Second, we study the runtime cost associated with this more expressive representation. Third, we evaluate whether flip maps exhibit sufficient reuse to make map interning effective. Finally, we validate the exponential-to-linear compression predicted in Section~\ref{subsec:compression} on the polarity-sensitive function family introduced there.

\subsection{Experimental Setup}
\label{subsec:setup}

We evaluate XBDD on nine circuits from the IWLS'93 benchmark suite \cite{ref_iwls93} and
compare it with CUDD \cite{ref_cudd}, a widely used and mature BDD package. The evaluated
circuits include \texttt{C1355}, \texttt{C1908}, \texttt{C3540},
\texttt{C499}, \texttt{C880}, \texttt{cm150a}, \texttt{comp},
\texttt{mux}, and \texttt{my\_adder}, covering circuits with different
sizes and different degrees of polarity-related structural sharing.

Both implementations construct BDDs from the same BLIF circuit
descriptions. To isolate the effect of the representation itself, dynamic
variable reordering is disabled, and both implementations use the same
fixed variable order. Both are pinned to the same CPU core to reduce
scheduling interference. Each experiment is repeated 50 times, and the
average result is reported.

We evaluate four primary metrics. \emph{Final nodes} measures the number
of nodes required to represent all output functions after construction and
therefore directly reflects representation compactness. \emph{Peak nodes}
measures the maximum number of nodes allocated during construction,
capturing the pressure caused by intermediate results. \emph{Memory}
reports the total memory occupied by the BDD package, including the node
pool, unique table, computed cache, and, for XBDD, the map pool.
\emph{Wall time} measures the complete BDD construction time.

\subsection{Overall Performance}
\label{subsec:overall}

Table~\ref{tab:overall} summarizes the overall results. The most important
observation is that XBDD consistently reduces space consumption, although
the magnitude of the reduction depends strongly on the structure of the
circuit.

\begin{table*}[!t]
	\caption{Overall comparison between CUDD and XBDD on IWLS'93.
		Ratios are XBDD/CUDD; smaller values are better.}
	\label{tab:overall}
	\centering
	\footnotesize
	\begin{tabular}{l|rrr|rrr|rrr|rrr}
		\hline
		\multirow{2}{*}{Circuit}
		& \multicolumn{3}{c|}{Time (ms)}
		& \multicolumn{3}{c|}{Peak Nodes}
		& \multicolumn{3}{c|}{Final Nodes}
		& \multicolumn{3}{c}{Memory (MB)}\\
		& CUDD & XBDD & Ratio
		& CUDD & XBDD & Ratio
		& CUDD & XBDD & Ratio
		& CUDD & XBDD & Ratio\\
		\hline
		C1355     & 42.99 & 60.39 & 1.41 & 187026  & 13312   & 0.07 & 45922  & 3257   & 0.07 & 16.23  & 8.94  & 0.55\\
		C1908     & 76.50 & 98.21 & 1.28 & 164542  & 44032   & 0.27 & 36007  & 12633  & 0.35 & 15.50  & 9.65  & 0.62\\
		C3540     & 2224.03 & 2505.77 & 1.13 & 2907590 & 2311168 & 0.79 & 604559 & 477635 & 0.79 & 162.60 & 120.81 & 0.74\\
		C499      & 35.35 & 49.11 & 1.39 & 145124  & 12288   & 0.08 & 45922  & 3257   & 0.07 & 14.86  & 8.92  & 0.60\\
		C880      & 620.41 & 712.46 & 1.15 & 1370502 & 1199104 & 0.87 & 346660 & 330622 & 0.95 & 79.68 & 60.98 & 0.77\\
		cm150a    & 83.73 & 108.00 & 1.29 & 231994 & 128000 & 0.55 & 131071 & 82701 & 0.63 & 17.86 & 11.30 & 0.63\\
		comp      & 272.75 & 206.98 & 0.76 & 673498 & 267264 & 0.40 & 458698 & 262141 & 0.57 & 40.60 & 15.71 & 0.39\\
		mux       & 95.01 & 115.80 & 1.22 & 229950 & 126976 & 0.55 & 131071 & 82701 & 0.63 & 17.80 & 11.26 & 0.63\\
		my\_adder & 219.44 & 212.12 & 0.97 & 656124 & 394240 & 0.60 & 327677 & 327677 & 1.00 & 32.07 & 18.39 & 0.57\\
		\hline
	\end{tabular}
\end{table*}

The reduction in final nodes directly demonstrates the additional sharing
enabled by variable-flip maps. The strongest examples are \texttt{C1355}
and \texttt{C499}, where XBDD requires only 7\% of the nodes used by CUDD,
corresponding to a reduction of approximately 93\%. XBDD also substantially
reduces the final node count on \texttt{C1908}, \texttt{cm150a},
\texttt{comp}, and \texttt{mux}. The benefit is smaller on \texttt{C880},
while \texttt{my\_adder} has the same final node count in both
representations.

This variation is expected. A flip map does not guarantee compression for
every Boolean function; it creates additional sharing only when
sub-functions differ primarily in the polarity of their inputs. Circuits
containing more such structures therefore benefit more strongly.

Peak node counts exhibit a similar trend, but reveal an additional
benefit. Even for \texttt{my\_adder}, where the final node count is
unchanged, XBDD reduces the peak node count from 656,124 to 394,240.
Therefore, variable-flip sharing can reduce intermediate structures during
construction even when it does not reduce the final representation.

The reduction in node count, together with the compact 16-byte node
representation described in Section~\ref{subsec:memory_impl}, translates
into lower overall memory consumption. Across all evaluated circuits,
XBDD uses only 39\%--77\% of the memory consumed by CUDD. This shows that
the additional map metadata does not eliminate the space benefit obtained
from structural sharing.

The itemized data also reveals another phenomenon: on smaller circuits, the memory of XBDD is dominated by the fixed-capacity computed cache---the initial fixed capacity of 262144 slots alone occupies about 8 MB. For example, on \texttt{C1355} and \texttt{C499}, although the node count drops by about 93\%, the memory only drops by about 40\%; the reason is that when the node pool is compressed to an extremely small size, the drop in memory touches the lower bound determined by fixed overheads such as the computed cache, and thus no longer decreases in the same proportion as the node count.

\subsection{Space--Time Tradeoff}
\label{subsec:tradeoff}

The main cost of XBDD is additional computation during construction.
Compared with CUDD, every recursive operation may require map membership
tests, symmetric differences, and map interning, and both the unique table and the computed table additionally require hashing over the maps. The key question is therefore whether the resulting space reduction justifies this runtime overhead.

Fig.~\ref{fig:tradeoff} summarizes this tradeoff by plotting, for each
circuit, the normalized wall time against the normalized memory
consumption. Both values are normalized to CUDD, so $(1,1)$ represents
the CUDD baseline. Points below the horizontal line at 1 consume less
memory, while points to the left of the vertical line at 1 are also
faster than CUDD.

\begin{figure*}[!t]
	\centering
	\begin{minipage}{0.32\textwidth}
		\centering
		\includegraphics[width=\linewidth]{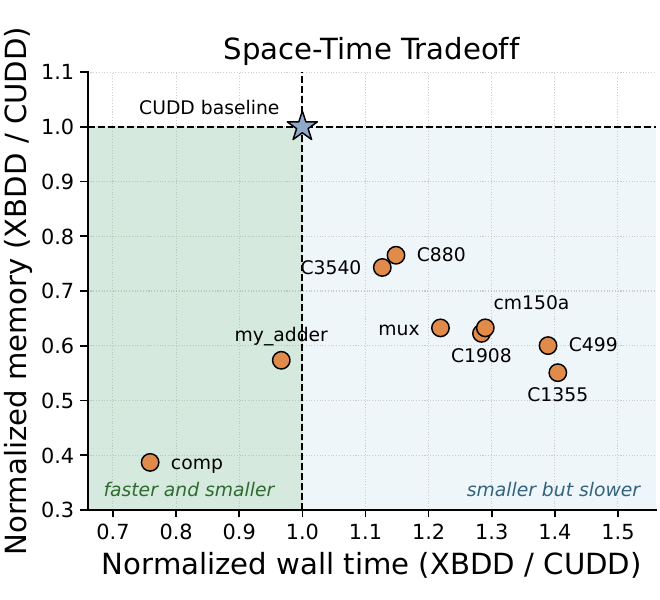}
		\caption{Space--time tradeoff of XBDD relative to CUDD.}
		\label{fig:tradeoff}
	\end{minipage}
	\hfill
	\begin{minipage}{0.32\textwidth}
		\centering
		\includegraphics[width=\linewidth]{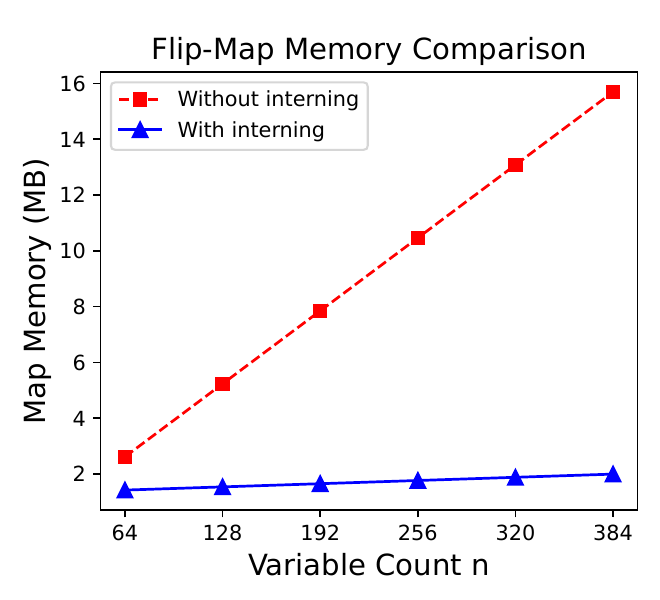}
		\caption{Total flip-map storage for C1355 as a function of the
			variable count $n$.}
		\label{fig:intern-space}
	\end{minipage}
	\hfill
	\begin{minipage}{0.32\textwidth}
		\centering
		\includegraphics[width=\linewidth]{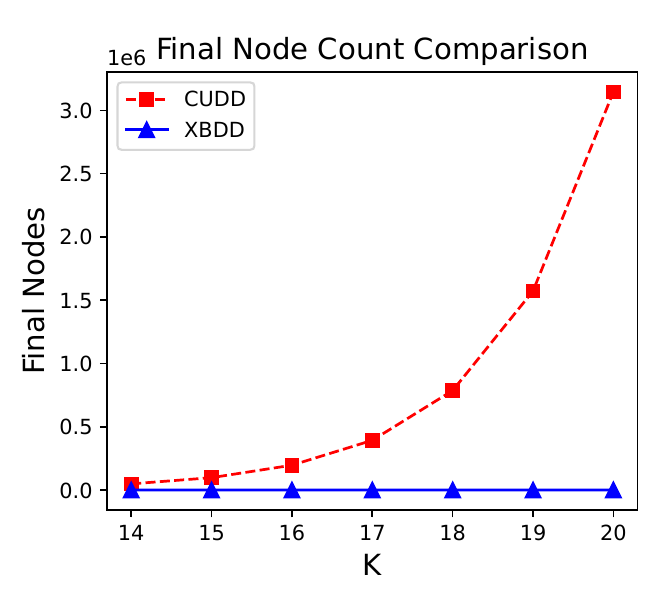}
		\caption{Final node count of XBDD and CUDD for the
			polarity-sensitive function family $F_k$.}
		\label{fig:expcompress}
	\end{minipage}
\end{figure*}

All evaluated circuits lie below the memory baseline, confirming that XBDD
consistently provides a net space saving. For smaller circuits such as
\texttt{C1355} and \texttt{C499}, this saving comes with a relatively
visible runtime overhead: XBDD is 41\% and 39\% slower, respectively.
However, these circuits also exhibit the strongest structural compression,
with their final node counts reduced by approximately 93\%.

The relative runtime overhead becomes smaller on larger circuits. For
example, XBDD is only 13\% slower on \texttt{C3540}, while reducing its
memory consumption by 26\%. The reason is that, as the construction grows,
operations such as node allocation, table expansion, and hash-table access
account for a larger fraction of the total runtime, reducing the relative
weight of map manipulation.

In some cases, structural compression can even compensate for the map
overhead. On \texttt{my\_adder}, XBDD and CUDD have nearly identical
construction times, while XBDD uses only 57\% of the memory. On
\texttt{comp}, XBDD is 24\% faster and uses only 39\% of the memory.
The reduction in intermediate nodes can decrease node-pool and table
expansions, partially or completely offsetting the additional cost of map
operations.

These results indicate that XBDD does not simply exchange memory for a
fixed runtime penalty. Instead, the tradeoff depends on the amount of
structural sharing exposed by the flip maps and on the scale of the BDD
construction. The runtime overhead is most visible on small circuits,
whereas larger constructions tend toward parity and may occasionally
benefit from the reduced number of intermediate nodes.

\subsection{Map Interning Efficiency}
\label{subsec:intern}

Map interning is intended to prevent the additional edge attribute from
becoming a new source of memory overhead. Its effectiveness depends on a
simple assumption: although maps may be referenced frequently, the number
of \emph{distinct} maps should remain relatively small.

To examine this property, we record the number of distinct maps stored in
the map pool, the total number of map references before deduplication, and
the dynamic map-pool hit rate. Table~\ref{tab:intern} summarizes the
results.

\begin{table}[!t]
	\caption{Compression and reuse achieved by map interning.}
	\label{tab:intern}
	\centering
	\footnotesize
	\begin{tabular}{l|r|r|r|r}
		\hline
		Circuit & Pool & Total & Comp. & Hit\\
		& Entries & Maps & Ratio & Rate\\
		\hline
		C1355     & 15024 & 342888 & 95.62\% & 98.27\%\\
		C1908     & 5159  & 197907 & 97.39\% & 98.84\%\\
		C3540     & 4754  & 575384 & 99.17\% & 99.85\%\\
		C499      & 15231 & 281071 & 94.58\% & 97.82\%\\
		C880      & 3309  & 312072 & 98.94\% & 99.84\%\\
		cm150a    & 303   & 480294 & 99.94\% & 99.96\%\\
		comp      & 34855 & 817152 & 95.73\% & 98.71\%\\
		mux       & 199   & 436628 & 99.95\% & 99.98\%\\
		my\_adder & 65567 & 645115 & 89.84\% & 96.24\%\\
		\hline
	\end{tabular}
\end{table}

The results show substantial map reuse across all evaluated circuits.
Interning eliminates between 89.84\% and 99.95\% of duplicate map
instances. In several circuits, hundreds of thousands of map references
collapse into only a few hundred or a few thousand distinct maps. For
example, \texttt{mux} contains 436,628 map references but only 199
distinct maps.

This high reuse rate directly supports the use of compact handles.
Instead of storing a complete bitmap at every occurrence, nodes and cache
entries store only a 32-bit map identifier, while each distinct bitmap is
stored once in the global pool.

The dynamic hit rate is also consistently high, ranging from 96.24\% to
99.98\%. Thus, most interning requests find an existing map rather than
creating a new one. The map pool consequently remains small, and map
interning introduces limited allocation and rehashing overhead.

The storage consequence of this reuse grows with the variable count.
Because a map is stored as an integral number of 64-bit words, its exact
size is $b=8\left\lceil\frac{n}{64}\right\rceil\ \text{bytes}$. Without interning, every one of the $M$ map references stores a full
bitmap, for $M\cdot b$ bytes in total. With interning, a reference
stores only a 32-bit handle and each of the $P$ distinct bitmaps is stored
once in the pool, for $4M + P\cdot b$ bytes.
Fig.~\ref{fig:intern-space} instantiates both expressions with the \texttt{C1355}
counters of Table~\ref{tab:intern} ($M = 342{,}888$, $P = 15{,}024$). The
uninterned cost grows with the full reference count and reaches 15.7\,MB
at $n = 384$, whereas the interned cost grows with the pool size only and
stays below 2.1\,MB, a $7.9\times$ reduction. Since the two slopes differ
by the factor $M/P$, the gap widens linearly in $n$: interning replaces a
per-reference dependence on the variable count with a constant 4-byte
handle, leaving the pool as the only $n$-dependent term.

Overall, these results confirm that flip maps are highly repetitive in
practice. Map interning therefore converts what could have been a
substantial per-edge storage cost into a small shared metadata structure.

\subsection{Exponential Compression Case Study}
\label{subsec:expcompress}

The benchmark results above demonstrate that variable-flip maps provide
substantial compression on practical circuits. We finally examine whether
they can also produce the asymptotic compression predicted in
Section~\ref{subsec:compression}.

Recall the function family
\begin{equation}
	F_k(c_1,\ldots,c_k,x_1,\ldots,x_k)
	=
	\bigwedge_{j=1}^{k}(x_j\oplus c_j),
	\label{eq:case-family}
\end{equation}
under the fixed variable order
\begin{equation}
	c_1\succ\cdots\succ c_k
	\succ x_1\succ\cdots\succ x_k.
\end{equation}
The $c_j$ variables determine the polarities of the $x_j$ variables.
A complement-edge-only BDD cannot directly share all these polarity
variants, whereas XBDD can encode their differences in edge maps.

We construct $F_k$ for $k=14,\ldots,20$ using both XBDD and CUDD, with
dynamic variable reordering disabled and the same fixed variable order.
Fig.~\ref{fig:expcompress} reports the resulting final node counts.

The two representations exhibit fundamentally different growth trends.
The CUDD node count approximately doubles as $k$ increases, consistent
with exponential growth. In contrast, the XBDD node count grows only
linearly.

At $k=20$, CUDD requires 3,145,725 nodes, whereas XBDD requires only
60 nodes. The difference is therefore not merely a constant-factor
reduction. The experiment confirms the structural result from
Section~\ref{subsec:compression}: for this function family, per-edge
variable-flip maps reduce the representation from exponential to linear
size.

This case study also explains why the gains on practical circuits vary.
Real circuits do not necessarily exhibit this polarity structure at every
level, but whenever similar local structures occur, XBDD can exploit the
same sharing mechanism. The practical benchmark results can therefore be
viewed as different degrees of this underlying structural opportunity.

\section{Related Work}
\label{sec:related}

\noindent\textbf{Decision-diagram representations.}
Since Bryant established ROBDDs as a canonical representation of Boolean
functions \cite{ref_bryant,ref_bryant_survey}, many decision-diagram
variants have been proposed to exploit different types of structure.
Zero-suppressed BDDs (ZBDDs) are particularly effective for representing
sparse sets \cite{ref_minato}, while multi-terminal BDDs extend terminal
values beyond Boolean constants \cite{ref_mtbdd}. Free BDDs relax the
global variable-order restriction \cite{ref_fbdd}, and sentential decision
diagrams exploit structured decompositions to provide a more general
canonical representation \cite{ref_sdd}. Decision diagrams have also been
extended to other domains. For example, QMDDs and TDDs represent quantum
states and circuits \cite{ref_qmdd,ref_tdd}, while LIMDDs and LimTDDs
attach local transformations to edges to expose additional equivalences
between quantum substructures \cite{ref_limdd,ref_limtdd}. These approaches
change the representation to exploit application-specific forms of
redundancy. XBDD instead remains within the ROBDD framework and targets a
specific redundancy that frequently appears in Boolean functions:
sub-functions that share the same structure but differ in the polarities
of some input variables.

\noindent\textbf{Edge-augmented BDDs.}
A particularly relevant line of research increases sharing by attaching
additional information to BDD edges. Complement edges are the most widely
used example: a single complement bit allows a node representing $f$ to
also represent $\overline{f}$, avoiding a separate graph for the
complementary function \cite{ref_brace}. Such edges have become a standard
component of mature BDD implementations. More general edge annotations
have also been explored. Typed edges and shared BDDs associate different
transformations with edges, including output inversion, input inversion,
and variable shifting \cite{ref_madre,ref_minato_sbdd}. Other approaches
encode reduction information on edges, such as chain-reduced BDDs,
tagged BDDs, and edge-specified reduction techniques
\cite{ref_bryant_chain,ref_tbdd,ref_esrbdd,ref_cesrbdd}.

XBDD is closely related to this general idea of moving structural
differences from nodes to edges, but focuses specifically on
\emph{local input-polarity differences}. A conventional complement edge
encodes one output transformation $f(x)\rightarrow\overline{f(x)}$ whereas the variable-flip map considered in this work describes transformations of the form $f(x)\rightarrow f(x\oplus p)$ where $p$ may specify an arbitrary subset of input variables. More
importantly, XBDD integrates this representation into the ROBDD
construction framework by providing normalization and cofactor operations
that preserve a canonical representation, together with a compact
implementation based on bitmap encoding and map interning. This enables
input-polarity variants to share the same underlying nodes while retaining
the constant-time equivalence property of canonical BDD representations.

\noindent\textbf{Efficient BDD implementations.}
A complementary body of work focuses on improving the efficiency of BDD
construction without fundamentally changing the represented Boolean
function. Mature packages such as CUDD and BuDDy employ techniques
including unique tables, computed caches, complement edges, reference
counting, and garbage collection \cite{ref_cudd,ref_buddy}. Subsequent
work has improved memory locality and data layout through
memory-hierarchy-aware and cache-conscious designs
\cite{ref_sanghavi,ref_long,ref_modernhw}, as well as pointerless
representations that reduce node-storage overhead \cite{ref_janssen_ptr}.
Parallel implementations exploit multi-core processors, GPUs, and
external-memory systems to scale BDD manipulation to larger problems
\cite{ref_sylvan,ref_gpu,ref_adiar}. Other work reduces the number of
temporary nodes generated during construction \cite{ref_krauss}.
Variable ordering is another orthogonal direction because it can
dramatically affect BDD size; dynamic techniques such as sifting search
for improved orders during construction \cite{ref_rudell}.

\section{Conclusion}
\label{sec:conclusion}

This paper presents XBDD, an extension of ROBDD that exploits
input-polarity redundancy through per-edge variable-flip maps.
While complement edges allow functions with opposite output polarities
to share the same nodes, XBDD further enables sub-functions that differ
only in selected input polarities to reuse the same underlying structure.
Normalization and map-aware cofactor operations guarantee constructive
uniqueness of the generated edges, while
bitmap encoding, map interning, and compact data structures keep the
additional overhead small.

Experiments on IWLS'93 circuits show that XBDD substantially reduces both
node count and memory consumption compared with CUDD, while introducing
a moderate runtime overhead that decreases as construction scale grows.
For polarity-sensitive functions, XBDD can further reduce the
representation from exponential to linear size. Future work will study
the interaction between variable-flip maps and dynamic variable ordering,
as well as their applicability to larger and parallel BDD workloads.

\section*{Acknowledgment}
During the preparation of this manuscript, we used ChatGPT, a generative AI tool, to assist with grammar checking and language polishing.

\bibliographystyle{IEEEtran}
\bibliography{refs}

\end{document}